\documentclass[10pt,a4paper]{elsarticle}
\usepackage[latin1]{inputenc}
\usepackage{amssymb}
\usepackage{amsfonts}
\usepackage{amsthm}
\usepackage{amsmath}
\usepackage{pstricks,pst-plot}
\usepackage{pstricks,pst-node,pst-plot}
\usepackage{graphicx}
\usepackage{textcomp}

\usepackage{epsfig}

\begin{document}
\begin{frontmatter}
\title{On bounding the difference between the maximum degree and the chromatic number by a constant}
\author[vw]{Vera Weil\corref{cor1}}
\ead{vera.weil@oms.rwth-aachen.de}

\author[os]{Oliver Schaudt}
\ead{schaudto@uni-koeln.de}

\cortext[cor1]{Corresponding author}
\address[vw]{Chair of Management Science, RWTH Aachen University, 
Kackertstr.~7, 52072 Aachen, Germany}
\address[os]{Department for Computer Science, University of Cologne, 
Weyertal 80, 50931 Cologne, Germany}

\begin{abstract}

We provide a finite forbidden induced subgraph characterization for the 
graph class~$\varUpsilon_k$, 
for all $k \in \mathbb{N}_0$, which is defined as follows. 
A graph is in $\varUpsilon_k$ if for any induced subgraph, 
$\Delta \leq \chi -1 + k$ holds,
where $\Delta$ is the maximum degree and $\chi$ is the chromatic number of the subgraph.

We compare these results 
with those given in~\cite{unserpaper}, 
where we studied the graph class $\varOmega_k$, 
for $k \in \mathbb{N}_0$, 
whose graphs are such that for any induced subgraph,  
$\Delta \leq \omega -1 + k$ holds, 
where $\omega$ denotes the clique number of a graph. 
In particular, 
we give a characterization in terms of $\varOmega_k$ and $\varUpsilon_k$ of those graphs where 
the neighborhood of every vertex is perfect. 
\end{abstract}

\end{frontmatter}

\theoremstyle{plain}
\newtheorem{theorem}{Theorem}
\newtheorem{lemma}{Lemma}
\newtheorem{observation}{Observation}
\newtheorem{proposition}{Proposition}
\newtheorem{corollary}{Corollary}

\section{Introduction}

A graph class $\mathcal G$ is called \emph{hereditary} if for every graph $G \in \mathcal G$, 
every induced subgraph of $G$ is also a member of $\mathcal G$. 
If we can describe a graph class $\mathcal G$ by excluding a (not necessarily finite) set of graphs as induced subgraphs, then this graph class is hereditary.
A \textit{clique} in a graph is a set of vertices of the graph that are pairwise adjacent.  
A maximal clique that is of largest size in a graph $G$ is called a \emph{maximum clique} of the graph. 
By $\omega(G)$ we denote the size of a maximum clique in a graph $G$. 
A \textit{coloring} of a graph is 
an assignment of colors to every vertex of the graph 
such that adjacent vertices do not receive the same color. 
A coloring that uses a minimum number of colors is called 
an \textit{optimal coloring}. 
The number of colors used in an optimal coloring of a graph 
$G$ is denoted by $\chi(G)$, the so called \textit{chromatic number}. 
With $V(G)$, 
we denote the vertex set of a graph $G$. 
The 
\textit{neighborhood} of a vertex $v \in V(G)$ is denoted by $N(v)$ 
and comprises all vertices adjacent to $v$. 
The \textit{degree} of $v$ corresponds to $|N(v)|$. 
Finally, $\Delta$ denotes the \textit{maximum degree} of a graph, 
that is, the maximum over all vertex degrees in a graph.

By Brook's Theorem~\cite{Brooks}, $\chi \leq \Delta + 1$ holds. 
However, it is not possible to give a \textit{lower} bound on $\chi$ in terms of $\Delta$ only.
By $K_{n,m}$, $n,m \in \mathbb{N}$, 
we denote the complete bipartite graph where one partition 
consists of $n$ and the other of $m$ vertices 
(for example, the $K_{1,4}$ can be found in Figure~\ref{figure F2 graphs chromatic}). 
The set of graphs $K_{1,p}$, for all $p \in \mathbb{N}$, yields an example for an infinite family of graphs where the 
difference between $\chi$ and $\Delta$ is $p-2$, and hence unbounded in this set.

Let $p \in \mathbb{N}$. 
With $J_p$, we denote the following graph.    
Consider a clique of size $p-1$ and a $K_{1,p}$. 
Let $a$ be a vertex in the p-partition of $K_{1,p}$ 
and let $b$ be a vertex in the clique. 
Add the edge $\{a,b\}$ (cf.~Figure~\ref{figure JP}).

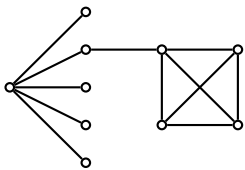
\begin{figure}[htb]
\begin{center}
\begin{pspicture}(5,2.5)(0,0)

	  \cnode(0.0,1.0){0.07}{A} 
	  \cnode(1.0,2.0){0.07}{D}
	  \cnode(1.0,1.5){0.07}{C}
 	  \cnode(1.0,1.0){0.07}{B} 
	  \cnode(1.0,0.5){0.07}{F}
	  \cnode(1.0,0.0){0.07}{E}
%
 	  \ncline[linestyle=solid]{A}{B}
 	  \ncline[linestyle=solid]{A}{C}
 	  \ncline[linestyle=solid]{A}{D}
 	  \ncline[linestyle=solid]{A}{E}
 	  \ncline[linestyle=solid]{A}{F}
 	  
 	  \cnode(2.0,1.5){0.07}{Q} 
	  \cnode(2.0,0.5){0.07}{R} 
	  \cnode(3.0,1.5){0.07}{S} 
	  \cnode(3.0,0.5){0.07}{T} 
	  
 	  \ncline[linestyle=solid]{Q}{R}
 	  \ncline[linestyle=solid]{Q}{S}
 	  \ncline[linestyle=solid]{Q}{T}
 	  \ncline[linestyle=solid]{R}{S}
 	  \ncline[linestyle=solid]{R}{T}
 	  \ncline[linestyle=solid]{S}{T}

 	  \ncline[linestyle=solid]{Q}{C}
\end{pspicture}
\end{center}
\caption{The graph $J_4$.}
\label{figure JP}
\end{figure}

In the resulting graph, 
$\Delta$ and $\chi$ equal $p+1$, 
hence the difference equals $0$, 
but the induced $K_{1,p}$ yields a graph where the difference is $p-2$. 
In other words, 
the value of the difference between the maximum degree and the chromatic number in the host graph 
is not necessairily an upper bound for the respective difference in an induced subgraph. 
This gives rise to the following question. 
Which graphs guarantee that for every induced subgraph, 
the difference between the maximum degree and the chromatic number is at most some given number $k$?

In Section~\ref{Sec2}, we answer the above question in the following way. 
For a fixed $k \in \mathbb{N}_0$, 
let $\varUpsilon_k$ be the class of graphs $G$ for which $\Delta(H) + 1\leq \chi(H) + k$ holds for all induced subgraphs $H$ of $G$. 
We provide a minimal forbidden induced subgraph characterization for $\varUpsilon_k$. 
Moreover, we prove that the order of the respective minimal forbidden induced subgraph set is finite.
Hence the problem of recognition of such graphs can be solved in polynomial time. 

In Section~\ref{section neighborhood perfect graphs}, 
we connect the graph classes considered in Section~\ref{Sec2} to the graph classes in~\cite{unserpaper}.   
In the latter, we considered $\varOmega_k$, where $G \in \varOmega_k$ if for all induced subgraphs $H$ of $G$, $\Delta(H) \leq \omega(H) + k$ holds. 

In Section~\ref{Sec4}, we close this paper by giving a short summary of the results and pointing out some future directions, 
including generalizations of the results presented here.   

Before we move on the Section~\ref{Sec2}, we give some further preliminary information. 
We distinguish between induced subgraphs and (partial) subgraphs. 
Since we deal with graph invariants, we are allowed to treat isomorphic graphs as identical.
For example, if a graph $G$ is an induced subgraph of a graph $H$ and $G$ is isomorphic to a graph $L$, then we say that $L$ is an induced subgraph of $H$.
A vertex is \textit{dominating} in a graph if it is adjacent to all other vertices of the graph.
In a coloring of a graph, a \textit{color class} is the set of all vertices to which the same color is assigned to. 

We denote the set of minimal forbidden induced subgraphs of $\varUpsilon_k$ by $F(\chi,k)$, 
In other words, $F \in F(\chi,k)$ if and only if $F \not \in \varUpsilon_k$ and all proper induced subgraphs of $F$ are contained in $\varUpsilon_k$. 
Observe that $G \in \varUpsilon_k$ if and only if $G$ is $F(\chi,k)$-free.

\section{The hereditary graph class $\varUpsilon_k$, $k \in \mathbb{N}_0$}\label{Sec2}


Let $k \in \mathbb{N}_0$. 
In~\cite{unserpaper}, 
the authors introduced a family of hereditary graph classes that is quite similar to $\varUpsilon_k$, 
namely $\varOmega_k$. 
By definition, 
$G \in \varOmega_k$ if every induced subgraph $H$ of $G$, including $G$ itself, 
obeys $\Delta(H) \leq \omega(H) + k - 1$. 
Let $F(\omega,k)$ denote the set of minimal forbidden induced subgraphs of $\varOmega_k$. 
We reword the characterization of $F(\omega,k)$ from~\cite{unserpaper} 
in~Theorem~\ref{charac of Fk}. 

\begin{theorem}[\cite{unserpaper}]\label{charac of Fk}
Let $G$ be a graph. 
$G \in F(\omega,k)$ if and only if the following conditions hold:
\begin{enumerate}
  \item $G$ has a unique dominating vertex $v$,  \label{first cond}
  \item the intersection of all maximum cliques of G contains solely $v$,  \label{second cond}
  \item $\Delta(G) = \omega(G)+k$ \label{third cond}.
\end{enumerate}
In particular, $\Delta(G) = |V(G)|-1$ and $\omega(G)=|V(G)| - k - 1$.
\end{theorem}

Our results are primarily based on Theorem~\ref{charac of Fk chromatic}, 
whose analogism to Theorem~\ref{charac of Fk} is obvious.

\begin{theorem}\label{charac of Fk chromatic}
Let $G$ be a graph. 
$G \in F(\chi,k)$ if and only if the following conditions hold:
\begin{enumerate}
	\item $G$ has a unique dominating vertex $v$, \label{first condition chromatic}
	\item each color class in every optimal coloring of $G-v$ consists of at least two vertices, \label{second condition chromatic}
	\item $\Delta(G)=\chi(G)+k$. \label{third condition chromatic}	
\end{enumerate}

\end{theorem}

\begin{proof}
Let $G \in F(\chi,k)$. 
Since $G$ is a minimal forbidden induced subgraph, 
all induced subgraphs of $G$ are contained in $\varUpsilon_k$ except for $G$ itself. 
Thus $\Delta(G) > \chi(G) + k-1$.
Choose a vertex $v$ of maximum degree in $G$ and 
let $H$ be the graph induced in $G$ by the vertex set $\{v\} \cup N(v)$. 
Observe that $H \subseteq G$ is not in $\varUpsilon_k$, since $\Delta(H)=\Delta(G) > \chi(G)+k-1 \geq \chi(H) + k-1$. 
Hence, $H \cong G$ by minimality of $G$, and thus, $G$ contains a dominating vertex, namely~$v$. 
Assume there exists a vertex $x \in V(G) \setminus \{v\}$ such that  
$\chi(G-x) = \chi(G)-1$. 
This equality in particular holds if $x$ is a dominating vertex. 
Then 
\begin{equation}\label{mini unique chromatic}
\Delta(G-x) = \Delta(G)-1 \geq \chi(G) + k - 1 = \chi(G-x) + k.  
\end{equation}
Thus $G-x \not \in \varUpsilon_k$, contradicting the minimality of $G$. 
Hence, Conditions~\ref{first condition chromatic} and~\ref{second condition chromatic} follow. 
Let $x \in N(v)$. 
Due to Condition~\ref{first condition chromatic}, 
the degree of $x$ is at most $\Delta(G)-2$. 
Hence, $\Delta(G-x)=\Delta(G)-1$.  
Due to Condition~\ref{second condition chromatic}, 
$\chi(G-x)=\chi(G)$. 
Assume $\Delta(G)\geq\chi(G)+k+1$. 
Then
$\Delta(G-x) = \Delta(G) - 1 \geq (\chi(G) + k + 1) - 1 = \chi(G-x) + k.$ 
That is, $G-x \not \in \varUpsilon_k$, a contradiction. 
Hence $\Delta(G) = \chi(G) + k$, and the third condition follows.

Let $G$ obey 
Conditions~\ref{first condition chromatic}, 
\ref{second condition chromatic} and \ref{third condition chromatic}.  
We have to prove that $G \in F(\chi,k)$. 
Since $$\Delta(G)=\chi(G)+k > \chi(G) + k - 1,$$ 
$G$ is a forbidden induced subgraph for every graph contained in $\varUpsilon_k$. 
To see that $G$ is minimal, 
let $L \in F(\chi,k)$ be an induced subgraph of $G$. 
Observe that the graph induced by $(V(L)\setminus \{y\}) \cup \{v\}$ 
is isomorphic to $L$, hence without loss of generality, let $y=v$. 
Let $S=V(G)\setminus V(L)$ and recall that $\Delta(G)=V(G)-1$.  
Thus, 
\begin{equation*}
\chi(L)+k=
V(L)-1=
V(G)-|S|-1=
\Delta(G)-|S|=\chi(G)+k-|S|. 
\end{equation*} 
That is, 
$\chi(G)-\chi(L)=|S|= |V(G)|-|V(L)|$. 
In other words, 
every vertex in $V(G)\setminus V(L)$ coincides with its own color class in every optimal coloring of $G$, 
yielding $S=\emptyset$. 
Hence, $G=L$. 
\end{proof}

\begin{corollary}\label{newcoro}
Let $G \in F(\chi,k)$. 
Then $\Delta(G) = |V(G)|-1$ and $\chi(G)=|V(G)| - k - 1$.
\end{corollary}
\begin{proof}
If Conditions~\ref{first condition chromatic}, \ref{second condition chromatic} and~\ref{third condition chromatic} from Theorem~\ref{charac of Fk chromatic} hold for $G$, 
then the dominating vertex $v$ has maximum degree, thus $\Delta(G)=|V(G)|-1$.
By Condition~\ref{third condition chromatic}, $\Delta(G) = \chi(G) + k$, and therefore $\chi(G)=|V(G)|-k-1$.
\end{proof}


Our next result, Proposition~\ref{proposition Delta leq 2k+2 chromatic}, provides a bound in terms of $k$ on the order of the minimal forbidden induced subgraphs of $\varUpsilon_k$. 
By $K_n$ we denote the complete graph on $n$ vertices. 

\begin{proposition}\label{proposition Delta leq 2k+2 chromatic}
Let $G \in F(\chi,k)$. 
Then $2\chi(G) - 2 \leq \Delta(G) \leq 2k+2$.
Moreover, $2 \leq \chi(G) \leq k+2$.
\end{proposition}

\begin{proof}
Let $G \in F(\chi,k)$. 
Observe that $K_1 \in \varUpsilon_k$ for all $k \geq 0$, hence $\chi(G) \geq 2$.
By Corollary~\ref{newcoro}, $|V(G)|=\Delta(G)+1$. 
Thus, Theorem~\ref{charac of Fk chromatic}, Condition~\ref{second condition chromatic}, 
yields 
$\Delta(G) \geq 2(\chi(G)-1).$
By Theorem~\ref{charac of Fk chromatic}, Condition~\ref{third condition chromatic}, 
$\chi(G)+k = \Delta(G)$ holds. 
Hence, $\chi(G)+k = \Delta(G) \geq 2(\chi(G)-1)$, 
leading to $\chi(G) \leq k+2$ and $\Delta(G) \leq 2k+2$. 
\end{proof}

Proposition~\ref{proposition Delta leq 2k+2 chromatic} has an important consequence: 
it yields a bound for the order of minimal forbidden induced subgraphs.
In other words, 
$F(\chi,k)$ is a subset of the set of graphs that have at most $2k+3$ vertices, 
and is therefore finite for every $k \in \mathbb{N}_0$.  

\begin{observation}
For every $k \in \mathbb{N}_0$, the set of minimal forbidden induced subgraphs of $\varUpsilon_k$ is finite. 
\end{observation} 

Hence, the problem of recognition of these graphs can be solved in polynomial time. 
Note that for any fixed $k \in \mathbb{N}_0$, the set $F(\omega,k)$ is also finite (cf.~\cite{unserpaper}).  
More similarities become clear when comparing the sets $\varOmega_k$ and~$\varUpsilon_k$. 
Recall that the chromatic number is always at least as large as the clique number. 

\begin{observation}\label{observation upsilon part of omega fixed k}
For every $k \in \mathbb{N}_0$, $\varOmega_k\subseteq \varUpsilon_k $.
\end{observation}

Observation~\ref{observation upsilon part of omega fixed k} does not necessarily imply $F(\chi,k) \subseteq F(\omega,k)$. 
Every graph in $F(\chi,k)$ is forbidden as a subgraph of a graph in $\varOmega_k$, 
but with regard to this property not necessarily minimal. 
However, as demonstrated by the next results, a graph in $F(\chi,k)$ 
is also in $F(\omega,k)$ 
if it is perfect. 
A graph $G$ is \textit{perfect} if  
for $G$ and all its induced subgraphs the clique number and the chromatic number coincide. 
The class of perfect graphs is among the best studied hereditary graph classes~(cf.~\cite{ Perfect.graphs,B.Toft,M.C.Golumbic}). 

\begin{lemma}\label{lemma perfect intersection coloring}
Let $G$ be a perfect graph. 
Then the intersection of all maximum cliques of $G$ is empty if and only if 
in every optimal coloring of $G$, 
every color class consists of at least~2 vertices. 
\end{lemma}

\begin{proof}
Let $G$ be a perfect graph. 
The intersection of all maximum cliques of $G$ is empty if and only if 
there exists no vertex $x \in V(G)$ such that $\omega(G-x)=\omega(G)-1$. 
This is equivalent to the statement that there exists no vertex $x \in V(G)$ such that $\chi(G-x)=\chi(G)-1$,  
since $G-x$ is also perfect. 
In other words, there is no vertex that forms its own color class in any optimal coloring of $G$.
\end{proof}

With this lemma, 
we can now formulate the following result. 

\begin{theorem}\label{theorem perfekt minforb}
Let $k \in \mathbb{N}_0$ and let $G$ be a perfect graph.
Then $G \in F(\omega,k)$ if and only if $G \in F(\chi,k)$. 
In particular, 
if all graphs in $F(\omega,k)$ are perfect, 
then $F(\omega,k)=F(\chi,k)$. 
\end{theorem}

\begin{proof}
Let $k \in \mathbb{N}_0$ and let $G$ be a perfect graph.
Observe that $G \in F(\omega,k)$ if and only if $G$ meets Conditions~\ref{first cond}, \ref{second cond} and~\ref{third cond} of Theorem~\ref{charac of Fk}. 
Obviously, Condition~\ref{first cond} of Theorem~\ref{charac of Fk} 
and Condition~\ref{first condition chromatic} of Theorem~\ref{charac of Fk chromatic} 
coincide. 
By Lemma~\ref{lemma perfect intersection coloring}, 
$G$ obeys Condition~\ref{second cond} of Theorem~\ref{charac of Fk} 
if and only 
if $G$ obeys Condition~\ref{second condition chromatic} of Theorem~\ref{charac of Fk chromatic}, 
since $G$ is perfect. 
Finally, 
$\Delta(G)=\omega(G)+k$ is equivalent to $\Delta(G)=\chi(G)+1$, 
again due to perfectness of $G$. 
All in all, 
$G \in F(\omega,k)$ if and only if $G \in F(\chi,k)$. 
\end{proof}
%

We remark that the statement of Theorem~\ref{theorem perfekt minforb} 
can be described by $F(\omega,k)\cap PG = F(\chi,k) \cap PG$, 
where $PG$ denotes the class of perfect graphs. 
With Theorem~\ref{theorem perfekt minforb}, 
it is easy to show that the sets 
$\varOmega_k$ and $\varUpsilon_k$ are the same for $k=0$ and $k=1$. For $s\in \mathbb{N}$, let $P_s$ denote the path on $s$ vertices. 

\begin{theorem}\label{upsilon for 0}
$F(\chi,0)=\{P_3\}=F(\omega,0)$. That is, $\varUpsilon_0$ consists of unions of complete graphs.
\end{theorem}

\begin{proof}
By Theorem~4 in~\cite{unserpaper}, $F(\omega,0)=P_3$. 
In particular, $F(\omega,0)$ consists of perfect graphs only.  
Theorem \ref{theorem perfekt minforb} completes the proof. 
\end{proof}

Theorem~\ref{charac of F1 chromatic} states that also 
$F(\omega,1)$ and $F(\chi,1)$ coincide. 
By $W_r$, $r \in \mathbb{N}$, $r \geq 3$, 
we denote the cycle on $r$ vertices where a dominating vertex is added. \label{definition wheel} 
For the graphs in $F(\chi,1)$, cf.~Figure~\ref{figure F1 graphs}.

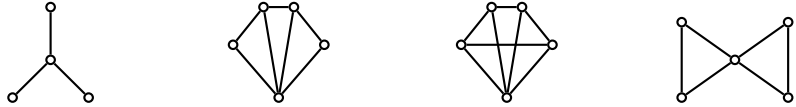
\begin{figure}[htb]
\begin{center}
\begin{pspicture}(11,2.5)(0,0)

	  \cnode(1.0,1.0){0.07}{ZK} 
	  \cnode(1.0,1.7){0.07}{KO}
	  \cnode(0.5,0.5){0.07}{KL}
	  \cnode(1.5,0.5){0.07}{KR} 
	  
 	  \ncline[linestyle=solid]{ZK}{KO}
 	  \ncline[linestyle=solid]{ZK}{KL}
 	  \ncline[linestyle=solid]{ZK}{KR}
	    
	  \cnode(4.0,0.5){0.07}{GZ} 
	  \cnode(3.4,1.2){0.07}{GL}
	  \cnode(3.8,1.7){0.07}{GLM}
	  \cnode(4.2,1.7){0.07}{GRM} 
	  \cnode(4.6,1.2){0.07}{GR} 
	
 	  \ncline[linestyle=solid]{GZ}{GL}
 	  \ncline[linestyle=solid]{GZ}{GLM}
 	  \ncline[linestyle=solid]{GZ}{GRM}
 	  \ncline[linestyle=solid]{GZ}{GR}
	    
	  \ncline[linestyle=solid]{GL}{GLM}
	  \ncline[linestyle=solid]{GLM}{GRM}
	  \ncline[linestyle=solid]{GRM}{GR}

	  \cnode(7.0,0.5){0.07}{WZ} 
	  \cnode(6.4,1.2){0.07}{WL}
	  \cnode(6.8,1.7){0.07}{WLM}
	  \cnode(7.2,1.7){0.07}{WRM} 
	  \cnode(7.6,1.2){0.07}{WR} 
	
 	  \ncline[linestyle=solid]{WZ}{WL}
 	  \ncline[linestyle=solid]{WZ}{WLM}
 	  \ncline[linestyle=solid]{WZ}{WRM}
 	  \ncline[linestyle=solid]{WZ}{WR}
	    
	  \ncline[linestyle=solid]{WL}{WLM}
	  \ncline[linestyle=solid]{WLM}{WRM}
	  \ncline[linestyle=solid]{WRM}{WR}
	  \ncline[linestyle=solid]{WR}{WL}
	    
	  \cnode(10.0,1.0){0.07}{BZ} 
	  \cnode(9.3,1.5){0.07}{BLO}
	  \cnode(9.3,0.5){0.07}{BLU}
	  \cnode(10.7,1.5){0.07}{BRO} 
	  \cnode(10.7,0.5){0.07}{BRU} 
	
 	  \ncline[linestyle=solid]{BZ}{BLO}
 	  \ncline[linestyle=solid]{BZ}{BLU}
 	  \ncline[linestyle=solid]{BZ}{BRO}
 	  \ncline[linestyle=solid]{BZ}{BRU}
	    
	  \ncline[linestyle=solid]{BLO}{BLU}
	  \ncline[linestyle=solid]{BRO}{BRU}

\end{pspicture}
\end{center}
\caption{The graphs claw, gem, $W_4$, butterfly.}
\label{figure F1 graphs}
\end{figure}

\begin{theorem}\label{charac of F1 chromatic}
$F(\chi,1) = \{\mbox{claw, gem, $W_4$, butterfly}\}=F(\omega,1)$. 
\end{theorem}

\begin{proof}
By Theorem~5 in~\cite{unserpaper}, $F(\omega,1) = \{\mbox{claw, gem, $W_4$, butterfly}\}$. 
Hence all graphs in $F(\omega,1)$ are perfect graphs. 
Theorem~\ref{theorem perfekt minforb} completes the proof. 
\end{proof}

The union of two graphs $G$ and $H$ is denoted by $G \cup H$. 
For $s\in \mathbb{N}$, let $C_s$ denote the cycle on $s$ vertices.
In order to compare the sets $F(\omega,2)$ and $F(\chi,2)$, 
we restate Theorem~6 of~\cite{unserpaper}, 
in a slightly adapted version that is based on the observation 
that every $K_3$-free supergraph of $K_2 \cup K_2 \cup K_1$ 
on five vertices is either a subgraph of $K_{2,3}$ or is the $C_5$-graph. 
The $K_6-3e$ is the $K_3$ where three pairwise non-incident edges are removed (it corresponds to the last graph in the last row of Figure~\ref{figure F2 graphs chromatic}
when the dominating vertex is removed). 
If we say subgraph respectively supergraph we allow both edges and vertices to be removed respectively added to the host graph.

\begin{theorem}[\cite{unserpaper}]\label{charac of F2}
Let $G$ be a graph. 
$G \in F(\omega,2)$ if and only if $G$ contains a dominating vertex $v$ and one of the following holds:
\begin{enumerate}
	\item $G-v \cong \overline{K_4}$, \label{Thm 6 of unserpaper K4}
	\item $G-v$ is a supergraph of $K_2 \cup K_2 \cup K_1$ and a subgraph of $K_{2,3}$,\label{Thm 6 unserpaper nicht C5}
	\item $G-v \cong C_5$, \label{Thm 6 of unserpaper C5}
	\item $G-v \cong S_3$, \label{Thm 6 of unserpaper S3}
\item $G-v$ is a supergraph of $K_3 \cup K_3$ and a subgraph of $K_{6}-3e$.\label{Thm 6 of unserpaper subsuper} 
	\end{enumerate}		
\end{theorem}

All graphs contained in $F(\chi,2)$ are 
shown in Figure~\ref{figure F2 graphs chromatic}.
With $C_5^{(3)}$ and $C_5^{(4)}$ we denote the graphs 
that correspond to a $C_5$ with a $K_1$ attached to three respectively four consecutive vertices of the $C_5$. 
Both graphs, drawn with a dominating vertex, 
can be found in Figure~\ref{figure F2 graphs chromatic}, 
namely the last two graphs in the second row.

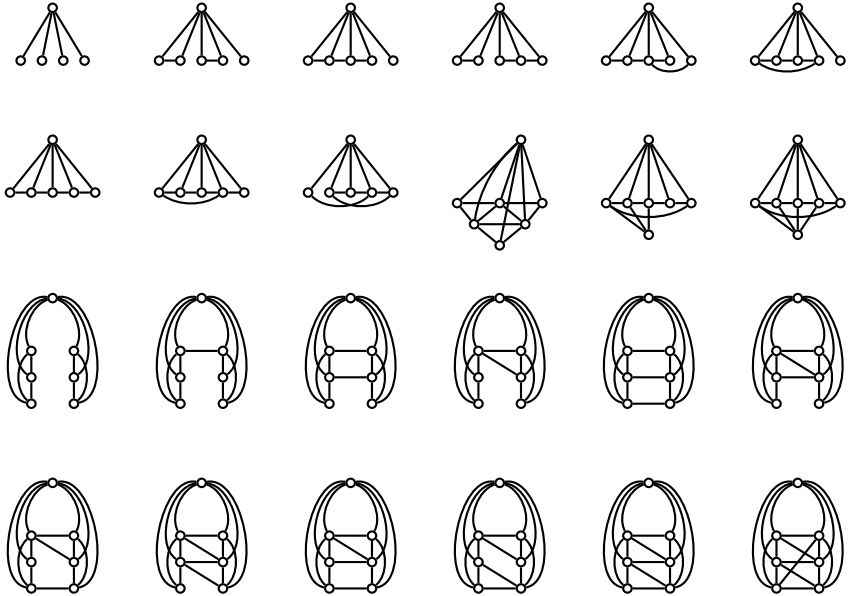
\begin{figure}[hbt]
\psset{xunit=0.28cm}
\psset{yunit=0.7cm}
\begin{pspicture}(40,12)
	  \cnode(3.0,11.4){0.07cm}{Z11} 
	  \cnode(1.5,10.4){0.07cm}{L11}
	  \cnode(2.5,10.4){0.07cm}{ML11}
	  \cnode(3.5,10.4){0.07cm}{MR11} 
	  \cnode(4.5,10.4){0.07cm}{R11} 
	  
	  \ncline[linestyle=solid]{Z11}{L11}
	  \ncline[linestyle=solid]{Z11}{ML11}
	  \ncline[linestyle=solid]{Z11}{MR11}
	  \ncline[linestyle=solid]{Z11}{R11}

	  \cnode(10.0,11.4){0.07cm}{Z12} 
	  \cnode(8.0,10.4){0.07cm}{L12}
	  \cnode(9.0,10.4){0.07cm}{ML12}
	  \cnode(10.0,10.4){0.07cm}{M12} 
	  \cnode(11.0,10.4){0.07cm}{MR12} 
	  \cnode(12.0,10.4){0.07cm}{R12}
	  
	  \ncline[linestyle=solid]{Z12}{L12}
	  \ncline[linestyle=solid]{Z12}{ML12}
	  \ncline[linestyle=solid]{Z12}{M12}
	  \ncline[linestyle=solid]{Z12}{MR12}
	  \ncline[linestyle=solid]{Z12}{R12}

	  \ncline[linestyle=solid]{L12}{ML12}
	  \ncline[linestyle=solid]{M12}{MR12}

	  \cnode(17.0,11.4){0.07cm}{Z13} 
	  \cnode(15.0,10.4){0.07cm}{L13}
	  \cnode(16.0,10.4){0.07cm}{ML13}
	  \cnode(17.0,10.4){0.07cm}{M13} 
	  \cnode(18.0,10.4){0.07cm}{MR13} 
	  \cnode(19.0,10.4){0.07cm}{R13}
	  
	  \ncline[linestyle=solid]{Z13}{L13}
	  \ncline[linestyle=solid]{Z13}{ML13}
	  \ncline[linestyle=solid]{Z13}{M13}
	  \ncline[linestyle=solid]{Z13}{MR13}
	  \ncline[linestyle=solid]{Z13}{R13}

	  \ncline[linestyle=solid]{L13}{ML13}
	  \ncline[linestyle=solid]{ML13}{M13}
	  \ncline[linestyle=solid]{M13}{MR13}

	  \cnode(24.0,11.4){0.07cm}{Z14} 
	  \cnode(22.0,10.4){0.07cm}{L14}
	  \cnode(23.0,10.4){0.07cm}{ML14}
	  \cnode(24.0,10.4){0.07cm}{M14} 
	  \cnode(25.0,10.4){0.07cm}{MR14} 
	  \cnode(26.0,10.4){0.07cm}{R14}
	  
	  \ncline[linestyle=solid]{Z14}{L14}
	  \ncline[linestyle=solid]{Z14}{ML14}
	  \ncline[linestyle=solid]{Z14}{M14}
	  \ncline[linestyle=solid]{Z14}{MR14}
	  \ncline[linestyle=solid]{Z14}{R14}

	  \ncline[linestyle=solid]{L14}{ML14}
	  \ncline[linestyle=solid]{M14}{MR14}
	  \ncline[linestyle=solid]{MR14}{R14}
	  
	  \cnode(31.0,11.4){0.07cm}{Z15} 
	  \cnode(29.0,10.4){0.07cm}{L15}
	  \cnode(30.0,10.4){0.07cm}{ML15}
	  \cnode(31.0,10.4){0.07cm}{M15} 
	  \cnode(32.0,10.4){0.07cm}{MR15} 
	  \cnode(33.0,10.4){0.07cm}{R15}
	  
	  \ncline[linestyle=solid]{Z15}{L15}
	  \ncline[linestyle=solid]{Z15}{ML15}
	  \ncline[linestyle=solid]{Z15}{M15}
	  \ncline[linestyle=solid]{Z15}{MR15}
	  \ncline[linestyle=solid]{Z15}{R15}

	  \ncline[linestyle=solid]{L15}{ML15}
	  \ncline[linestyle=solid]{ML15}{M15}
	  \ncline[linestyle=solid]{M15}{MR15}
	  \ncarc[arcangle=-50]{-}{M15}{R15}

	  \cnode(38.0,11.4){0.07cm}{Z16} 
	  \cnode(36.0,10.4){0.07cm}{L16}
	  \cnode(37.0,10.4){0.07cm}{ML16}
	  \cnode(38.0,10.4){0.07cm}{M16} 
	  \cnode(39.0,10.4){0.07cm}{MR16} 
	  \cnode(40.0,10.4){0.07cm}{R16}
	  
	  \ncline[linestyle=solid]{Z16}{L16}
	  \ncline[linestyle=solid]{Z16}{ML16}
	  \ncline[linestyle=solid]{Z16}{M16}
	  \ncline[linestyle=solid]{Z16}{MR16}
	  \ncline[linestyle=solid]{Z16}{R16}

	  \ncline[linestyle=solid]{L16}{ML16}
	  \ncline[linestyle=solid]{ML16}{M16}
	  \ncline[linestyle=solid]{M16}{MR16}
	  \ncarc[arcangle=-35]{-}{L16}{MR16}

	  \cnode(3.0,8.9){0.07cm}{Z22} 
	  \cnode(1.0,7.9){0.07cm}{L22}
	  \cnode(2.0,7.9){0.07cm}{ML22}
	  \cnode(3.0,7.9){0.07cm}{M22} 
	  \cnode(4.0,7.9){0.07cm}{MR22} 
	  \cnode(5.0,7.9){0.07cm}{R22}
	  
	  \ncline[linestyle=solid]{Z22}{L22}
	  \ncline[linestyle=solid]{Z22}{ML22}
	  \ncline[linestyle=solid]{Z22}{M22}
	  \ncline[linestyle=solid]{Z22}{MR22}
	  \ncline[linestyle=solid]{Z22}{R22}

	  \ncline[linestyle=solid]{L22}{ML22}
	  \ncline[linestyle=solid]{ML22}{M22}
	  \ncline[linestyle=solid]{M22}{MR22}
	  \ncline[linestyle=solid]{MR22}{R22}
	  
%
%

	  \cnode(10.0,8.9){0.07cm}{Z22} 
	  \cnode(8.0,7.9){0.07cm}{L22}
	  \cnode(9.0,7.9){0.07cm}{ML22}
	  \cnode(10.0,7.9){0.07cm}{M22} 
	  \cnode(11.0,7.9){0.07cm}{MR22} 
	  \cnode(12.0,7.9){0.07cm}{R22}
	  
	  \ncline[linestyle=solid]{Z22}{L22}
	  \ncline[linestyle=solid]{Z22}{ML22}
	  \ncline[linestyle=solid]{Z22}{M22}
	  \ncline[linestyle=solid]{Z22}{MR22}
	  \ncline[linestyle=solid]{Z22}{R22}

	  \ncline[linestyle=solid]{L22}{ML22}
	  \ncline[linestyle=solid]{ML22}{M22}
	  \ncline[linestyle=solid]{M22}{MR22}
	  \ncline[linestyle=solid]{MR22}{R22}
	  \ncarc[arcangle=-35]{-}{L22}{MR22}

	  \cnode(17.0,8.9){0.07cm}{Z23} 
	  \cnode(15.0,7.9){0.07cm}{L23}
	  \cnode(16.0,7.9){0.07cm}{ML23}
	  \cnode(17.0,7.9){0.07cm}{M23} 
	  \cnode(18.0,7.9){0.07cm}{MR23} 
	  \cnode(19.0,7.9){0.07cm}{R23}
	  
	  \ncline[linestyle=solid]{Z23}{L23}
	  \ncline[linestyle=solid]{Z23}{ML23}
	  \ncline[linestyle=solid]{Z23}{M23}
	  \ncline[linestyle=solid]{Z23}{MR23}
	  \ncline[linestyle=solid]{Z23}{R23}

	  \ncline[linestyle=solid]{ML23}{M23}
	  \ncline[linestyle=solid]{M23}{MR23}
	  \ncline[linestyle=solid]{MR23}{R23}
	  \ncarc[arcangle=-45]{-}{L23}{MR23}
	  \ncarc[arcangle=-45]{-}{ML23}{R23}
	  
	  \cnode(25.0,8.9){0.07cm}{Z24}
	  
	  \cnode(22.0,7.7){0.07cm}{OL24} 
	  \cnode(24.0,7.7){0.07cm}{OM24}
	  \cnode(26.0,7.7){0.07cm}{OR24}
	  \cnode(22.8,7.3){0.07cm}{ML24} 
	  \cnode(25.2,7.3){0.07cm}{MR24} 
	  \cnode(24.0,6.9){0.07cm}{U24} 
	  
	  \ncline[linestyle=solid]{Z24}{OL24}
	  \ncline[linestyle=solid]{Z24}{OM24}
	  \ncline[linestyle=solid]{Z24}{OR24}
	  \ncline[linestyle=solid]{Z24}{MR24}
	  \ncline[linestyle=solid]{Z24}{U24}
	  
	  \ncline[linestyle=solid]{OL24}{OM24}
	  \ncline[linestyle=solid]{OM24}{OR24}
	  \ncline[linestyle=solid]{OR24}{MR24}
	  \ncline[linestyle=solid]{MR24}{U24}
	  \ncline[linestyle=solid]{U24}{ML24}
	  \ncline[linestyle=solid]{ML24}{OL24}
	  
	  \ncline[linestyle=solid]{OM24}{ML24}
	  \ncline[linestyle=solid]{OM24}{MR24}
	  \ncline[linestyle=solid]{MR24}{ML24}
	  \ncarc[arcangle=20]{-}{ML24}{Z24}

	  \cnode(31.0,8.9){0.07cm}{Z25}

	  \cnode(29.0,7.7){0.07cm}{L25} 
	  \cnode(30.0,7.7){0.07cm}{LM25}
	  \cnode(31.0,7.7){0.07cm}{M25}
	  \cnode(32.0,7.7){0.07cm}{RM25} 
	  \cnode(33.0,7.7){0.07cm}{R25} 
	  \cnode(31.0,7.1){0.07cm}{U25} 

	  \ncline[linestyle=solid]{Z25}{L25}
	  \ncline[linestyle=solid]{Z25}{LM25}
	  \ncline[linestyle=solid]{Z25}{M25}
	  \ncline[linestyle=solid]{Z25}{R25}
	  \ncline[linestyle=solid]{Z25}{RM25}

	  \ncline[linestyle=solid]{L25}{LM25}
	  \ncline[linestyle=solid]{LM25}{M25}
	  \ncline[linestyle=solid]{M25}{RM25}
	  \ncline[linestyle=solid]{RM25}{R25}
	  \ncarc[arcangle=-35]{-}{L25}{R25}

	  \ncline[linestyle=solid]{L25}{U25}
	  \ncline[linestyle=solid]{LM25}{U25}
	  \ncline[linestyle=solid]{M25}{U25}

	  \cnode(38.0,8.9){0.07cm}{Z26}

	  \cnode(36.0,7.7){0.07cm}{L26} 
	  \cnode(37.0,7.7){0.07cm}{LM26}
	  \cnode(38.0,7.7){0.07cm}{M26}
	  \cnode(39.0,7.7){0.07cm}{RM26} 
	  \cnode(40.0,7.7){0.07cm}{R26} 
	  \cnode(38.0,7.1){0.07cm}{U26} 
	  
	  \ncline[linestyle=solid]{Z26}{L26}
	  \ncline[linestyle=solid]{Z26}{LM26}
	  \ncline[linestyle=solid]{Z26}{M26}
	  \ncline[linestyle=solid]{Z26}{R26}
	  \ncline[linestyle=solid]{Z26}{RM26}

	  \ncline[linestyle=solid]{L26}{LM26}
	  \ncline[linestyle=solid]{LM26}{M26}
	  \ncline[linestyle=solid]{M26}{RM26}
	  \ncline[linestyle=solid]{RM26}{R26}
	  \ncarc[arcangle=-35]{-}{L26}{R26}

	  \ncline[linestyle=solid]{L26}{U26}
	  \ncline[linestyle=solid]{LM26}{U26}
	  \ncline[linestyle=solid]{M26}{U26}
	  \ncline[linestyle=solid]{RM26}{U26}

	  \cnode(3.0,5.9){0.07cm}{Z31} 
	  \cnode(2.0,4.9){0.07cm}{LO31}
	  \cnode(2.0,4.4){0.07cm}{LM31}
	  \cnode(2.0,3.9){0.07cm}{LU31} 
	  
	  \cnode(4.0,4.9){0.07cm}{RO31} 
	  \cnode(4.0,4.4){0.07cm}{RM31}
	  \cnode(4.0,3.9){0.07cm}{RU31}
	  
	  \ncarc[arcangle=-90]{-}{Z31}{LU31}
	  \ncarc[arcangle=-70]{-}{Z31}{LM31}
	  \ncarc[arcangle=-50]{-}{Z31}{LO31}

	  \ncarc[arcangle=90]{-}{Z31}{RU31}
	  \ncarc[arcangle=70]{-}{Z31}{RM31}
	  \ncarc[arcangle=50]{-}{Z31}{RO31}
	  
	  \ncline[linestyle=solid]{LO31}{LM31}
	  \ncline[linestyle=solid]{LU31}{LM31}
	  \ncarc[arcangle=-50]{-}{LO31}{LU31}
	  
	  \ncline[linestyle=solid]{RO31}{RM31}
	  \ncline[linestyle=solid]{RU31}{RM31}
	  \ncarc[arcangle=50]{-}{RO31}{RU31}
	
	  \cnode(10.0,5.9){0.07cm}{Z32} 
	  \cnode(9.0,4.9){0.07cm}{LO32}
	  \cnode(9.0,4.4){0.07cm}{LM32}
	  \cnode(9.0,3.9){0.07cm}{LU32} 
	  
	  \cnode(11.0,4.9){0.07cm}{RO32} 
	  \cnode(11.0,4.4){0.07cm}{RM32}
	  \cnode(11.0,3.9){0.07cm}{RU32}
	  
	  \ncarc[arcangle=-90]{-}{Z32}{LU32}
	  \ncarc[arcangle=-70]{-}{Z32}{LM32}
	  \ncarc[arcangle=-50]{-}{Z32}{LO32}

	  \ncarc[arcangle=90]{-}{Z32}{RU32}
	  \ncarc[arcangle=70]{-}{Z32}{RM32}
	  \ncarc[arcangle=50]{-}{Z32}{RO32}
	  
	  \ncline[linestyle=solid]{LO32}{LM32}
	  \ncline[linestyle=solid]{LU32}{LM32}
	  \ncarc[arcangle=-50]{-}{LO32}{LU32}
	  
	  \ncline[linestyle=solid]{RO32}{RM32}
	  \ncline[linestyle=solid]{RU32}{RM32}
	  \ncarc[arcangle=50]{-}{RO32}{RU32}

	  \ncline[linestyle=solid]{LO32}{RO32}

	  \cnode(17.0,5.9){0.07cm}{Z33} 
	  \cnode(16.0,4.9){0.07cm}{LO33}
	  \cnode(16.0,4.4){0.07cm}{LM33}
	  \cnode(16.0,3.9){0.07cm}{LU33} 
	  
	  \cnode(18.0,4.9){0.07cm}{RO33} 
	  \cnode(18.0,4.4){0.07cm}{RM33}
	  \cnode(18.0,3.9){0.07cm}{RU33}
	  
	  \ncarc[arcangle=-90]{-}{Z33}{LU33}
	  \ncarc[arcangle=-70]{-}{Z33}{LM33}
	  \ncarc[arcangle=-50]{-}{Z33}{LO33}

	  \ncarc[arcangle=90]{-}{Z33}{RU33}
	  \ncarc[arcangle=70]{-}{Z33}{RM33}
	  \ncarc[arcangle=50]{-}{Z33}{RO33}
	  
	  \ncline[linestyle=solid]{LO33}{LM33}
	  \ncline[linestyle=solid]{LU33}{LM33}
	  \ncarc[arcangle=-50]{-}{LO33}{LU33}
	  
	  \ncline[linestyle=solid]{RO33}{RM33}
	  \ncline[linestyle=solid]{RU33}{RM33}
	  \ncarc[arcangle=50]{-}{RO33}{RU33}
	  
	  \ncline[linestyle=solid]{LO33}{RO33}
	  \ncline[linestyle=solid]{LM33}{RM33}
	  
	  \cnode(24.0,5.9){0.07cm}{Z34} 
	  \cnode(23.0,4.9){0.07cm}{LO34}
	  \cnode(23.0,4.4){0.07cm}{LM34}
	  \cnode(23.0,3.9){0.07cm}{LU34} 
	  
	  \cnode(25.0,4.9){0.07cm}{RO34} 
	  \cnode(25.0,4.4){0.07cm}{RM34}
	  \cnode(25.0,3.9){0.07cm}{RU34}
	  
	  \ncarc[arcangle=-90]{-}{Z34}{LU34}
	  \ncarc[arcangle=-70]{-}{Z34}{LM34}
	  \ncarc[arcangle=-50]{-}{Z34}{LO34}

	  \ncarc[arcangle=90]{-}{Z34}{RU34}
	  \ncarc[arcangle=70]{-}{Z34}{RM34}
	  \ncarc[arcangle=50]{-}{Z34}{RO34}
	  
	  \ncline[linestyle=solid]{LO34}{LM34}
	  \ncline[linestyle=solid]{LU34}{LM34}
	  \ncarc[arcangle=-50]{-}{LO34}{LU34}
	  
	  \ncline[linestyle=solid]{RO34}{RM34}
	  \ncline[linestyle=solid]{RU34}{RM34}
	  \ncarc[arcangle=50]{-}{RO34}{RU34}

	  \ncline[linestyle=solid]{LO34}{RO34}
	  \ncline[linestyle=solid]{LO34}{RM34}
	  
	  \cnode(31.0,5.9){0.07cm}{Z35} 
	  \cnode(30.0,4.9){0.07cm}{LO35}
	  \cnode(30.0,4.4){0.07cm}{LM35}
	  \cnode(30.0,3.9){0.07cm}{LU35} 
	  
	  \cnode(32.0,4.9){0.07cm}{RO35} 
	  \cnode(32.0,4.4){0.07cm}{RM35}
	  \cnode(32.0,3.9){0.07cm}{RU35}
	  
	  \ncarc[arcangle=-90]{-}{Z35}{LU35}
	  \ncarc[arcangle=-70]{-}{Z35}{LM35}
	  \ncarc[arcangle=-50]{-}{Z35}{LO35}

	  \ncarc[arcangle=90]{-}{Z35}{RU35}
	  \ncarc[arcangle=70]{-}{Z35}{RM35}
	  \ncarc[arcangle=50]{-}{Z35}{RO35}
	  
	  \ncline[linestyle=solid]{LO35}{LM35}
	  \ncline[linestyle=solid]{LU35}{LM35}
	  \ncarc[arcangle=-50]{-}{LO35}{LU35}
	  
	  \ncline[linestyle=solid]{RO35}{RM35}
	  \ncline[linestyle=solid]{RU35}{RM35}
	  \ncarc[arcangle=50]{-}{RO35}{RU35}

	  \ncline[linestyle=solid]{LO35}{RO35}
	  \ncline[linestyle=solid]{LM35}{RM35}
	  \ncline[linestyle=solid]{LU35}{RU35}
	  
	  \cnode(38.0,5.9){0.07cm}{Z36} 
	  \cnode(37.0,4.9){0.07cm}{LO36}
	  \cnode(37.0,4.4){0.07cm}{LM36}
	  \cnode(37.0,3.9){0.07cm}{LU36} 
	  
	  \cnode(39.0,4.9){0.07cm}{RO36} 
	  \cnode(39.0,4.4){0.07cm}{RM36}
	  \cnode(39.0,3.9){0.07cm}{RU36}
	  
	  \ncarc[arcangle=-90]{-}{Z36}{LU36}
	  \ncarc[arcangle=-70]{-}{Z36}{LM36}
	  \ncarc[arcangle=-50]{-}{Z36}{LO36}

	  \ncarc[arcangle=90]{-}{Z36}{RU36}
	  \ncarc[arcangle=70]{-}{Z36}{RM36}
	  \ncarc[arcangle=50]{-}{Z36}{RO36}
	  
	  \ncline[linestyle=solid]{LO36}{LM36}
	  \ncline[linestyle=solid]{LU36}{LM36}
	  \ncarc[arcangle=-50]{-}{LO36}{LU36}
	  
	  \ncline[linestyle=solid]{RO36}{RM36}
	  \ncline[linestyle=solid]{RU36}{RM36}
	  \ncarc[arcangle=50]{-}{RO36}{RU36}

	  \ncline[linestyle=solid]{LO36}{RO36}
	  \ncline[linestyle=solid]{LM36}{RM36}
	  \ncline[linestyle=solid]{LO36}{RM36}
	  
	  \cnode(3.0,2.4){0.07cm}{Z41} 
	  \cnode(2.0,1.4){0.07cm}{LO41}
	  \cnode(2.0,0.9){0.07cm}{LM41}
	  \cnode(2.0,0.4){0.07cm}{LU41} 
	  
	  \cnode(4.0,1.4){0.07cm}{RO41} 
	  \cnode(4.0,0.9){0.07cm}{RM41}
	  \cnode(4.0,0.4){0.07cm}{RU41}
	  
	  \ncarc[arcangle=-90]{-}{Z41}{LU41}
	  \ncarc[arcangle=-70]{-}{Z41}{LM41}
	  \ncarc[arcangle=-50]{-}{Z41}{LO41}

	  \ncarc[arcangle=90]{-}{Z41}{RU41}
	  \ncarc[arcangle=70]{-}{Z41}{RM41}
	  \ncarc[arcangle=50]{-}{Z41}{RO41}
	  
	  \ncline[linestyle=solid]{LO41}{LM41}
	  \ncline[linestyle=solid]{LU41}{LM41}
	  \ncarc[arcangle=-50]{-}{LO41}{LU41}
	  
	  \ncline[linestyle=solid]{RO41}{RM41}
	  \ncline[linestyle=solid]{RU41}{RM41}
	  \ncarc[arcangle=50]{-}{RO41}{RU41}
	  
	  \ncline[linestyle=solid]{LO41}{RO41}
	  \ncline[linestyle=solid]{LO41}{RM41}
	  \ncline[linestyle=solid]{LU41}{RU41}

	  \cnode(10.0,2.4){0.07cm}{Z42} 
	  \cnode(9.0,1.4){0.07cm}{LO42}
	  \cnode(9.0,0.9){0.07cm}{LM42}
	  \cnode(9.0,0.4){0.07cm}{LU42} 
	  
	  \cnode(11.0,1.4){0.07cm}{RO42} 
	  \cnode(11.0,0.9){0.07cm}{RM42}
	  \cnode(11.0,0.4){0.07cm}{RU42}
	  
	  \ncarc[arcangle=-90]{-}{Z42}{LU42}
	  \ncarc[arcangle=-70]{-}{Z42}{LM42}
	  \ncarc[arcangle=-50]{-}{Z42}{LO42}

	  \ncarc[arcangle=90]{-}{Z42}{RU42}
	  \ncarc[arcangle=70]{-}{Z42}{RM42}
	  \ncarc[arcangle=50]{-}{Z42}{RO42}
	  
	  \ncline[linestyle=solid]{LO42}{LM42}
	  \ncline[linestyle=solid]{LU42}{LM42}
	  \ncarc[arcangle=-50]{-}{LO42}{LU42}
	  
	  \ncline[linestyle=solid]{RO42}{RM42}
	  \ncline[linestyle=solid]{RU42}{RM42}
	  \ncarc[arcangle=50]{-}{RO42}{RU42}
	  
	  \ncline[linestyle=solid]{LO42}{RO42}
	  \ncline[linestyle=solid]{LO42}{RM42}
	  \ncline[linestyle=solid]{LM42}{RM42}
	  \ncline[linestyle=solid]{LM42}{RU42}

	  \cnode(17.0,2.4){0.07cm}{Z43} 
	  \cnode(16.0,1.4){0.07cm}{LO43}
	  \cnode(16.0,0.9){0.07cm}{LM43}
	  \cnode(16.0,0.4){0.07cm}{LU43} 
	  
	  \cnode(18.0,1.4){0.07cm}{RO43} 
	  \cnode(18.0,0.9){0.07cm}{RM43}
	  \cnode(18.0,0.4){0.07cm}{RU43}
	  
	  \ncarc[arcangle=-90]{-}{Z43}{LU43}
	  \ncarc[arcangle=-70]{-}{Z43}{LM43}
	  \ncarc[arcangle=-50]{-}{Z43}{LO43}

	  \ncarc[arcangle=90]{-}{Z43}{RU43}
	  \ncarc[arcangle=70]{-}{Z43}{RM43}
	  \ncarc[arcangle=50]{-}{Z43}{RO43}
	  
	  \ncline[linestyle=solid]{LO43}{LM43}
	  \ncline[linestyle=solid]{LU43}{LM43}
	  \ncarc[arcangle=-50]{-}{LO43}{LU43}
	  
	  \ncline[linestyle=solid]{RO43}{RM43}
	  \ncline[linestyle=solid]{RU43}{RM43}
	  \ncarc[arcangle=50]{-}{RO43}{RU43}

	  \ncline[linestyle=solid]{LO43}{RO43}
	  \ncline[linestyle=solid]{LO43}{RM43}
	  \ncline[linestyle=solid]{LM43}{RM43}
	  \ncline[linestyle=solid]{LU43}{RU43}
	  
	  \cnode(24.0,2.4){0.07cm}{Z44} 
	  \cnode(23.0,1.4){0.07cm}{LO44}
	  \cnode(23.0,0.9){0.07cm}{LM44}
	  \cnode(23.0,0.4){0.07cm}{LU44} 
	  
	  \cnode(25.0,1.4){0.07cm}{RO44} 
	  \cnode(25.0,0.9){0.07cm}{RM44}
	  \cnode(25.0,0.4){0.07cm}{RU44}
	  
	  \ncarc[arcangle=-90]{-}{Z44}{LU44}
	  \ncarc[arcangle=-70]{-}{Z44}{LM44}
	  \ncarc[arcangle=-50]{-}{Z44}{LO44}

	  \ncarc[arcangle=90]{-}{Z44}{RU44}
	  \ncarc[arcangle=70]{-}{Z44}{RM44}
	  \ncarc[arcangle=50]{-}{Z44}{RO44}
	  
	  \ncline[linestyle=solid]{LO44}{LM44}
	  \ncline[linestyle=solid]{LU44}{LM44}
	  \ncarc[arcangle=-50]{-}{LO44}{LU44}
	  
	  \ncline[linestyle=solid]{RO44}{RM44}
	  \ncline[linestyle=solid]{RU44}{RM44}
	  \ncarc[arcangle=50]{-}{RO44}{RU44}

	  \ncline[linestyle=solid]{LO44}{RO44}
	  \ncline[linestyle=solid]{LO44}{RM44}
	  \ncline[linestyle=solid]{LM44}{RU44}
	  \ncline[linestyle=solid]{LU44}{RU44}

	  \cnode(31.0,2.4){0.07cm}{Z45} 
	  \cnode(30.0,1.4){0.07cm}{LO45}
	  \cnode(30.0,0.9){0.07cm}{LM45}
	  \cnode(30.0,0.4){0.07cm}{LU45} 
	  
	  \cnode(32.0,1.4){0.07cm}{RO45} 
	  \cnode(32.0,0.9){0.07cm}{RM45}
	  \cnode(32.0,0.4){0.07cm}{RU45}
	  
	  \ncarc[arcangle=-90]{-}{Z45}{LU45}
	  \ncarc[arcangle=-70]{-}{Z45}{LM45}
	  \ncarc[arcangle=-50]{-}{Z45}{LO45}

	  \ncarc[arcangle=90]{-}{Z45}{RU45}
	  \ncarc[arcangle=70]{-}{Z45}{RM45}
	  \ncarc[arcangle=50]{-}{Z45}{RO45}
	  
	  \ncline[linestyle=solid]{LO45}{LM45}
	  \ncline[linestyle=solid]{LU45}{LM45}
	  \ncarc[arcangle=-50]{-}{LO45}{LU45}
	  
	  \ncline[linestyle=solid]{RO45}{RM45}
	  \ncline[linestyle=solid]{RU45}{RM45}
	  \ncarc[arcangle=50]{-}{RO45}{RU45}

	  \ncline[linestyle=solid]{LO45}{RO45}
	  \ncline[linestyle=solid]{LO45}{RM45}
	  \ncline[linestyle=solid]{LM45}{RM45}
	  \ncline[linestyle=solid]{LM45}{RU45}
	  \ncline[linestyle=solid]{LU45}{RU45}

	  \cnode(38.0,2.4){0.07cm}{Z46} 
	  \cnode(37.0,1.4){0.07cm}{LO46}
	  \cnode(37.0,0.9){0.07cm}{LM46}
	  \cnode(37.0,0.4){0.07cm}{LU46} 
	  
	  \cnode(39.0,1.4){0.07cm}{RO46} 
	  \cnode(39.0,0.9){0.07cm}{RM46}
	  \cnode(39.0,0.4){0.07cm}{RU46}
	  
	  \ncarc[arcangle=-90]{-}{Z46}{LU46}
	  \ncarc[arcangle=-70]{-}{Z46}{LM46}
	  \ncarc[arcangle=-50]{-}{Z46}{LO46}

	  \ncarc[arcangle=90]{-}{Z46}{RU46}
	  \ncarc[arcangle=70]{-}{Z46}{RM46}
	  \ncarc[arcangle=50]{-}{Z46}{RO46}
	  
	  \ncline[linestyle=solid]{LO46}{LM46}
	  \ncline[linestyle=solid]{LU46}{LM46}
	  \ncarc[arcangle=-50]{-}{LO46}{LU46}
	  
	  \ncline[linestyle=solid]{RO46}{RM46}
	  \ncline[linestyle=solid]{RU46}{RM46}
	  \ncarc[arcangle=50]{-}{RO46}{RU46}
	  \ncline[linestyle=solid]{LO46}{RO46}
	  \ncline[linestyle=solid]{LO46}{RM46}
	  \ncline[linestyle=solid]{LM46}{RM46}
	  \ncline[linestyle=solid]{LM46}{RU46}
	  \ncline[linestyle=solid]{LU46}{RU46}
	  \ncline[linestyle=solid]{LU46}{RO46}

\end{pspicture}
\caption{The set $F(\chi,2)$.} 
\label{figure F2 graphs chromatic}
\end{figure}

\begin{theorem}\label{charac of F2 chromatic}
Let $G$ be a graph. 
Then $G\in F(\chi,2)$ if and only if $G$ contains a dominating vertex $v$ and one of the following holds:
\begin{enumerate}
 \item $G-v \cong \overline{K_4}$, \label{claim 4 vertices chi 1}
 \item $G-v$ is a supergraph of $K_2\cup K_2\cup K_1$ and a subgraph of $K_{2,3}$, \label{claim 5 vertices chi 2}
 \item $G-v$ consists of 6 vertices and is a subgraph of $K_6-3e$ such that one of the following holds: \label{claim 6 vertices chi 3}
 \begin{enumerate}
  \item $G-v \cong S_3$, \label{thisisthesun}
  \item $G-v$ is a supergraph of $K_3 \cup K_3$,\label{2 disj triangles}
  \item $G-v \cong C_5^{(3)}$, \label{c5 3 vertices} 
    \item $G-v \cong C_5^{(4)}$.\label{c5 4 vertices} 
  \end{enumerate}
\end{enumerate}
In particular, $F(\omega,2)\setminus \{W_5\} = F(\chi,2)\setminus \{C_5^{(3)}, C_5^{(4)}\}$. 
\end{theorem}

\begin{proof}
With Condition~\ref{first condition chromatic}, 
\ref{second condition chromatic} and \ref{third condition chromatic},  
we refer to the conditions stated in Theorem~\ref{charac of Fk chromatic}. 
With $[i.]$ we refer to Condition $i$ listed in Theorem~\ref{charac of F2 chromatic}.

Let $G$ be a graph with a unique dominating vertex $v$. 
Let $G-v$ obey $[\ref{claim 4 vertices chi 1}.]$, 
$[\ref{claim 5 vertices chi 2}.]$, $[\ref{thisisthesun}.]$ or $[\ref{2 disj triangles}.]$. 
Then $G-v$ is a perfect graph and 
obeys Condition~\ref{Thm 6 of unserpaper K4}, 
\ref{Thm 6 unserpaper nicht C5}, 
\ref{Thm 6 of unserpaper S3} or~\ref{Thm 6 of unserpaper subsuper} 
of Theorem~\ref{charac of F2}, 
respectively. 
By Theorem~\ref{theorem perfekt minforb}, 
$G \in F(\chi,2)$. 
If $G-v$ obeys $[\ref{c5 3 vertices}.]$ or $[\ref{c5 4 vertices}.]$, 
then it is easy to see that $G-v$ obeys Condition~\ref{first condition chromatic}, \ref{second condition chromatic} and \ref{third condition chromatic} of Theorem~\ref{charac of Fk chromatic}, 
and that hence, $G \in F(\chi,2)$. 

To show the reverse direction, let $G$ be a graph in $F(\chi,2)$. 
Observe that by Condition~\ref{first condition chromatic} of 
Theorem~\ref{charac of Fk chromatic}, 
$G$ has a unique dominating vertex, say $v$. 
By Theorem~\ref{theorem perfekt minforb}, 
$F(\chi,2) \cap PG = F(\omega,2) \cap PG$. 
In other words, 
if $G$ is a graph in $F(\chi,2)$ that is perfect, 
then and only then $G$ is a graph in $F(\omega,2)$ that is perfect. 
Therefore, 
$G$ obeys 
$[\ref{claim 4 vertices chi 1}.]$, 
$[\ref{claim 5 vertices chi 2}.]$, $[\ref{thisisthesun}.]$ 
or~$[\ref{2 disj triangles}.]$. 
Let now $G$ be a non-perfect graph. 
By Proposition~\ref{proposition Delta leq 2k+2 chromatic}, we have $2 \leq \chi(G) \leq 4$, therefore $|G|=\Delta(G)+1$ equals 5, 6 or 7 and 
thus $|G-v|=4$, $5$ or $6$ and $\chi(G-v)=1$, $2$ or $3$, respectively.  
Non-perfectness implies $\chi(G-v)= 3$ and therefore $|G-v|=6$. 
The only odd hole respectively anti-hole that can be embedded as induced 
subgraph in $G-v$ is therefore $C_5$. 
Hence, let $C$ be an induced $C_5$ in $G-v$ and let 
$u \in G-v$ be the vertex not in $C$. 
Since $v$ is a unique dominating vertex, 
$u$ is adjacent to at most four vertices of $C_5$. 
Moreover, 
if $u$ is adjacent to at most two vertices of $C_5$, 
or to three vertices of $C_5$ that are not consecutively ordered, 
then we always find a coloring of $G-v$ where one vertex forms a singleton color class, 
contradicting Condition~\ref{first condition chromatic} of 
Theorem~\ref{charac of Fk chromatic}. 
Hence, 
$G \cong C_5^{(3)}$ or 
$G \cong C_5^{(4)}$. 
By checking the three conditions listed in 
Theorem~\ref{charac of Fk chromatic}, 
is easy to see that both these graphs are in $F(\chi,2)$. 
This completes the proof.  
\end{proof}

To sum up, $F(\omega,0)=F(\chi,0)$, $F(\omega,1)=F(\chi,1)$, 
but 
$$F(\omega,2)\setminus \{W_5\}=F(\chi,2)\setminus \{C_5^{(3)}, C_5^{(4)}\}.$$ 
Observe that both  $C_5^{(3)}$ and $C_5^{(4)}$ contain $W_5$ as induced subgraph. 
The question arises what separates the set $F(\chi,k)$ from the set $F(\omega,k)$ for a fixed $k \in \mathbb{N}$. 
In order to answer this question, 
we will genereralize the result for $k=2$ in the following section. 

Before we proceed, 
observe that $|F(\chi,0)|=1$, $|F(\chi,1)|=4$ and $|F(\chi,2)|=24$. 
Moreover, $|F(\chi,3)|=402$ and $|F(\chi,4)|=25788$ (cf.~\cite{Jan.Personal}), 
hence, although finite, the sets of minimal forbidden induced subgraphs 
seem to grow very quickly compared to the increase of $k$.  
All minimal forbidden induced subgraphs for $k=1$, $2$ and $3$ 
can be downloaded from \textit{House of Graphs}~\cite{house.of.graphs.dam} by searching for the keywords ``maximum degree * chromatic number'' or 
``chi(G) + k'' where $k=1$, $2$ or $3$.

\section{Neighborhood perfect graphs}\label{section neighborhood perfect graphs}

By the Strong Perfect Graph Theorem~\cite{Perfect.graphs}, 
a graph is perfect if and only if it is free of 
\textit{odd holes} and \textit{odd anti-holes}. 
An odd hole is a cycle of odd length with at least five vertices. 
An odd anti-hole is the complement of an odd hole. 
We say that a graph is \emph{neighborhood perfect}  
if the neighborhood of every vertex is perfect. 
Observe that a graph is neighborhood perfect if and only if it does not contain an odd hole with a dominating vertex, 
or an odd anti-hole with a dominating vertex.

\begin{lemma}\label{lemma neighborhood Ffree}
Let $G$ be a neighborhood perfect graph and let $k \in \mathbb{N}_0$.  
$G$ is $F(\omega,k)$-free if and only if 
$G$ is $F(\chi,k)$-free. 
\end{lemma}

\begin{proof}
Let $G$ be a neighborhood perfect graph 
and let $H$ be an induced subgraph of $G$ with a dominating vertex, say $v$. 
Since $H$ is neighborhood perfect, $H-v$ is perfect. 
Let $k \in \mathbb{N}_0$ be such that $H \in F(\omega,k)$ 
or \mbox{$H \in F(\chi,k)$}. 
In this case, by Theorem~\ref{theorem perfekt minforb}, 
$H \in F(\chi,k)$ or $H \in F(\omega,k)$, respectively. 
\end{proof}

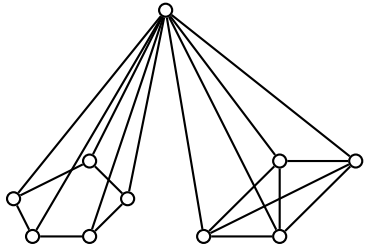
\begin{figure}
\centering
\begin{pspicture}(8,5) 
	  \cnode(3.5,4){0.1cm}{A}
	  
	  \cnode(1.5,1.5){0.1cm}{B} 
	  \cnode(1.75,1){0.1cm}{C}
	  \cnode(2.5,1){0.1cm}{D}
	  \cnode(3.0,1.5){0.1cm}{E} 
	  \cnode(2.5,2){0.1cm}{F} 
	  
	  \ncline[linestyle=solid]{A}{B}
	  \ncline[linestyle=solid]{A}{C}
	  \ncline[linestyle=solid]{A}{D}
	  \ncline[linestyle=solid]{A}{E}
	  \ncline[linestyle=solid]{A}{F}
	  
	  \ncline[linestyle=solid]{F}{B}
	  \ncline[linestyle=solid]{B}{C}
	  \ncline[linestyle=solid]{C}{D}
	  \ncline[linestyle=solid]{D}{E}
	  \ncline[linestyle=solid]{E}{F}
	  
	  \cnode(5,2){0.1cm}{G} 
	  \cnode(4,1){0.1cm}{H} 
	  \cnode(6,2){0.1cm}{I}
	  \cnode(5,1){0.1cm}{J}
	  
	  \ncline[linestyle=solid]{A}{G}
	  \ncline[linestyle=solid]{A}{H}
	  \ncline[linestyle=solid]{A}{I}
	  \ncline[linestyle=solid]{A}{J}
	  
	  \ncline[linestyle=solid]{G}{H}
	  \ncline[linestyle=solid]{G}{I}
	  \ncline[linestyle=solid]{G}{J}
	  \ncline[linestyle=solid]{H}{I}
	  \ncline[linestyle=solid]{H}{J}
	  \ncline[linestyle=solid]{I}{J}

\end{pspicture}
\caption{The opposite of Lemma~\ref{lemma neighborhood Ffree} is not true. 
In other words, we can not ommit the subgraph condition of Theorem \ref{neighborhood perfect theorem}.}
\label{neighborhood perfect G}
\end{figure}

Note that the opposite of Lemma~\ref{lemma neighborhood Ffree} is not true.  
Consider for example the graph drawn in Figure~\ref{neighborhood perfect G}, 
that is, the union of a $K_4$ and a $C_5$, where a dominating vertex is added. 
This graph, say $G$, is not neighborhood perfect, since $W_5$ is an induced subgraph of $G$. 
But, since $\Delta(G)=9$, $\omega(G)=5$ and $\chi(G)=5$, 
$G \in \varOmega_k$ and $G \in \varUpsilon_k$ for all $k \geq 5$, 
and $G \not \in \varOmega_k$ and $G \not \in \varUpsilon_k$ for $k \leq 4$. 
That is, $G \in \varOmega_k$ if and only if $G \in \varUpsilon_k$, for all $k \in \mathbb{N}_0$. 
In other words, 
$G$ is not neighborhood perfect, 
but for all $k \in \mathbb{N}_0$, $G$ is $F(\omega,k)$-free 
if and only if $G$ is $F(\chi,k)$-free. 

The proof of Theorem~\ref{neighborhood perfect theorem} needs a short preparation. 
The complement of a graph $G$ is denoted by $\overline{G}$. 
Consider the graph $W_{2r+3}$, $r\geq 1$ and observe that 
$W_{2r+3} \in F(\omega, 2r)$.  
Since $\chi(W_{2r+3})=4$, $W_{2r+3} \not \in F(\chi,2r)$.
Let further $B_{2r+3}$, $r \geq 1$, 
denote the anti-hole with $2r+3$ vertices where a dominating vertex is added. 
Observe that $B_{2r+3} \in F(\omega,r+1)$. 
It is easy to see that $\chi(B_{2r+3})=r+3$ and $\Delta(B_{2r+3})=2r+3$, 
hence $B_{2r+3} \not \in F(\chi,r+1)$.


\begin{observation}
Let $k \in \mathbb{N}$, $k \geq 2$. 
Then $W_{k+3} \in F(\omega,k) \setminus F(\chi,k)$ and 
$B_{2k+1} \in F(\omega,k) \setminus F(\chi,k)$. 
\end{observation}

We are now in the position to state Theorem~\ref{neighborhood perfect theorem}. 

\begin{theorem}\label{neighborhood perfect theorem}
 Let $G$ be a graph.
 Then the following statements are equivalent:
 \begin{enumerate}
  \item $G$ is a neighborhood perfect graph.
  \item For all $k \in \mathbb{N}_0$ and all induced subgraphs $H$ of $G$, 
	$ H \in \varOmega_k$ if and only if $H\in \varUpsilon_k$.
 \end{enumerate}
\end{theorem}

\begin{proof}
 Let $G$ be a neighborhood perfect graph and let $H$ be an induced subgraph of $G$. 
Hence, $H$ is also neighborhood perfect.  
 By Lemma~\ref{lemma neighborhood Ffree}, for every $k \in \mathbb{N}_0$, 
 $H$ is $F(\omega,k)$-free if and only if $H$ is $F(\chi,k)$-free. 
 
 Let on the other hand $G$ be a graph that is not neighborhood perfect. 
 Then, for some $k \geq 1$, 
 $G$ contains a subgraph, say $H$, such that $H \cong W_{k+3}$ or $H \cong B_{2k+1}$ 
 and hence a graph that is contained in $F(\omega,k)$, but not in $F(\chi,k)$. 
 Hence, $H \in \varUpsilon_k\setminus \varOmega_k$. 
 This completes the proof. 
 \end{proof}

Note that this result does not imply a polynomial algorithm for the recognition of 
neighborhood perfect graphs. 
However, due to~\cite{Perfect.graphs.recognition}, in polynomial time it is possible to check if the neighborhood of every vertex in a graph is perfect. 

\section{Final remarks}\label{Sec4}

We introduced the graph classes $\varUpsilon_k$, 
for all $k \in \mathbb{N}_0$. 
A member of $\varUpsilon_k$ has the property that 
all its induced subgraphs comply with 
$\Delta \leq \chi + k - 1$. 
We showed that those graphs can be characterized by a finite set of minimal forbidden induced subgraphs. 

Moreover, we found some relations to the results presented 
in~\cite{unserpaper}, 
where we require $\Delta \leq \omega + k - 1$ for every induced subgraph. 
In particular, the neighborhood perfect graphs are exactly those graphs 
for which $\varUpsilon_k$ and $\varOmega_k$ 
coincide for all $k \in \mathbb{N}_0$ and all induced subgraphs of the graph.

A future direction might 
restrict the graph classes $\varUpsilon_k$ to 
some graph universe like the claw-free graphs. 
There, a minimal forbidden induced subgraph $G$  
has a unique dominating vertex $v$, $\Delta(G)=\chi(G) + k$ 
but in every optimal coloring of $G-v$, 
every color class contains \textit{exactly} two vertices. 
This might lead to interesting results concerning the structure 
of minimal forbidden induced subgraphs. 
Furthermore, the question arises what happens when we compare $\varUpsilon_k$ to $\varOmega_j$
for some $j \in \mathbb{N}_0$, where $j \not = k$.
Finally, it might be of interest to focus on further 
graph parameters.  
We currently try to adapt our methods to the 
complement graph parameters of $\omega$ and $\chi$, 
that is, replacing $\omega$ or $\chi$ by the maximum size of an independet set or the clique cover number of the graph. 
In particular, 
we try to generalize our results, 
focussing on monotone graph parameters, 
where in our understanding, 
a parameter $\phi$ is monotone if for 
every \textit{induced} subgraph $H$ of some graph $G$, 
$\phi(H)\leq \phi(G)$. 

\section{Acknowledgement} 

We want to thank both reviewers for their very helpful and encouraging comments. Moreover, we want to thank Jan Goedgebeur from Ghent University, Belgium,  
for his very helpful computation of the sets $F(\chi,k)$, 
for small $k$. 
The reader interested in 
those sets 
is invited to have a look at 
\textit{House of Graphs}~(cf.~\cite{house.of.graphs.dam}).

\end{document}